\documentclass[twocolumn,10pt]{IEEEtran}

\usepackage{amsmath,amssymb,amsthm}
\usepackage{graphicx}
\usepackage{xcolor}
\usepackage{subfigure}
\usepackage[acronym]{glossaries}
\usepackage{tikz}
\usetikzlibrary{arrows.meta}

\definecolor{myred}{HTML}{B85450}
\definecolor{myblue}{HTML}{6C8EBF}

\newacronym{iid}{i.i.d.}{independent and identically distributed}
\newacronym{awgn}{AWGN}{additive white Gaussian noise}
\newacronym{snr}{SNR}{signal-to-noise ratio}
\newacronym{rf}{RF}{radio frequency}
\newacronym{milac}{MiLAC}{microwave linear analog computer}
\newacronym{ris}{RIS}{reconfigurable intelligent surface}
\newacronym{star-ris}{STAR-RIS}{simultaneously transmitting and reflecting RIS}
\newacronym{bd-ris}{BD-RIS}{beyond-diagonal RIS}
\newacronym{mimo}{MIMO}{multiple-input multiple-output}
\newacronym{em}{EM}{electromagnetic}

\newtheorem{proposition}{Proposition}

\begin{document}
\bstctlcite{BSTcontrol}

\title{\huge Microwave Linear Analog Computer (MiLAC) for\\Simultaneous Active and Passive Beamforming}

\author{Matteo~Nerini,~\IEEEmembership{Senior~Member,~IEEE},
        Bruno~Clerckx,~\IEEEmembership{Fellow,~IEEE}

\thanks{This work has been supported in part by UKRI under Grant EP/Y004086/1, EP/X040569/1, EP/Y037197/1, EP/X04047X/1, EP/Y037243/1.}
\thanks{The authors are with the Department of Electrical and Electronic Engineering, Imperial College London, SW7 2AZ London, U.K. (e-mail: m.nerini20@imperial.ac.uk; b.clerckx@imperial.ac.uk).}}

\maketitle

\begin{abstract}
\Glspl{milac} have recently emerged to enable high-performance and efficient beamforming in the analog domain.
In this paper, we introduce a dual-functionality framework for \gls{milac}-aided transceivers.
Beyond analog-domain precoding/combining (active beamforming), a \gls{milac} and its antenna array can simultaneously act as a \gls{ris} (passive beamforming).
This allows the \gls{milac} to execute beamforming for transmission/reception while reflecting external incident signals.
We provide an optimal reconfiguration strategy for this dual-functional \gls{milac}, and characterize the fundamental limits on the trade-off between active and passive rate, namely the capacity region bounds and the sum-rate capacity.
\end{abstract}

\glsresetall

\begin{IEEEkeywords}
Beamforming, microwave linear analog computer (MiLAC), reconfigurable intelligent surface (RIS).
\end{IEEEkeywords}

\section{Introduction}

\Gls{milac} has been proposed to enable high-performance beamforming purely in the analog domain by processing the signals via a reconfigurable microwave network \cite{ner25-2}.
In a \gls{milac}-aided transmitter, the \gls{milac} receives as input the transmitted symbols from the \gls{rf} chains, precodes them in the analog domain, and returns the precoded signals to the antennas.
In a \gls{milac}-aided receiver, the \gls{milac} receives the inputs from the receiving antennas, combines them in the analog domain, and sends the processed signals to the \gls{rf} chains for baseband detection.
This analog architecture offers a scalable solution towards large-scale \gls{mimo} for single-user communications \cite{ner25-3}, multi-user communications \cite{fan26,wu26,zho26,zha26}, and radar sensing \cite{liu26}.

Reconfigurable microwave networks have also been widely investigated to implement \glspl{ris} \cite{wu19,wu21}.
By dynamically adjusting their \gls{em} properties, \glspl{ris} can control the reflection of incident signals through passive beamforming.
\glspl{ris} have been implemented conventionally with reconfigurable loads connected to antenna elements.
Nevertheless, they can also be implemented with more general microwave networks, such as those in \gls{milac}, and the resulting \gls{ris} architectures are denoted as \gls{star-ris} \cite{mu22} and \gls{bd-ris} \cite{she22}.

In this work, we bridge \gls{milac} and \gls{ris} by arguing that a \gls{milac} mounted on a transceiver device can perform a second function in addition to precoding/combining.
When the antenna array connected to the \gls{milac} is sufficiently large, it can also serve as a \gls{bd-ris}, reflecting incident signals coming from external transmitting devices toward their intended receivers.
Consequently, we introduce a dual-functionality framework where a \gls{milac} can simultaneously execute active beamforming (precoding/combining signals from/to its \gls{rf} chains) and passive beamforming (reflecting external incident signals to desired directions).

\begin{figure*}[t]
\centering
\includegraphics[width=0.64\textwidth]{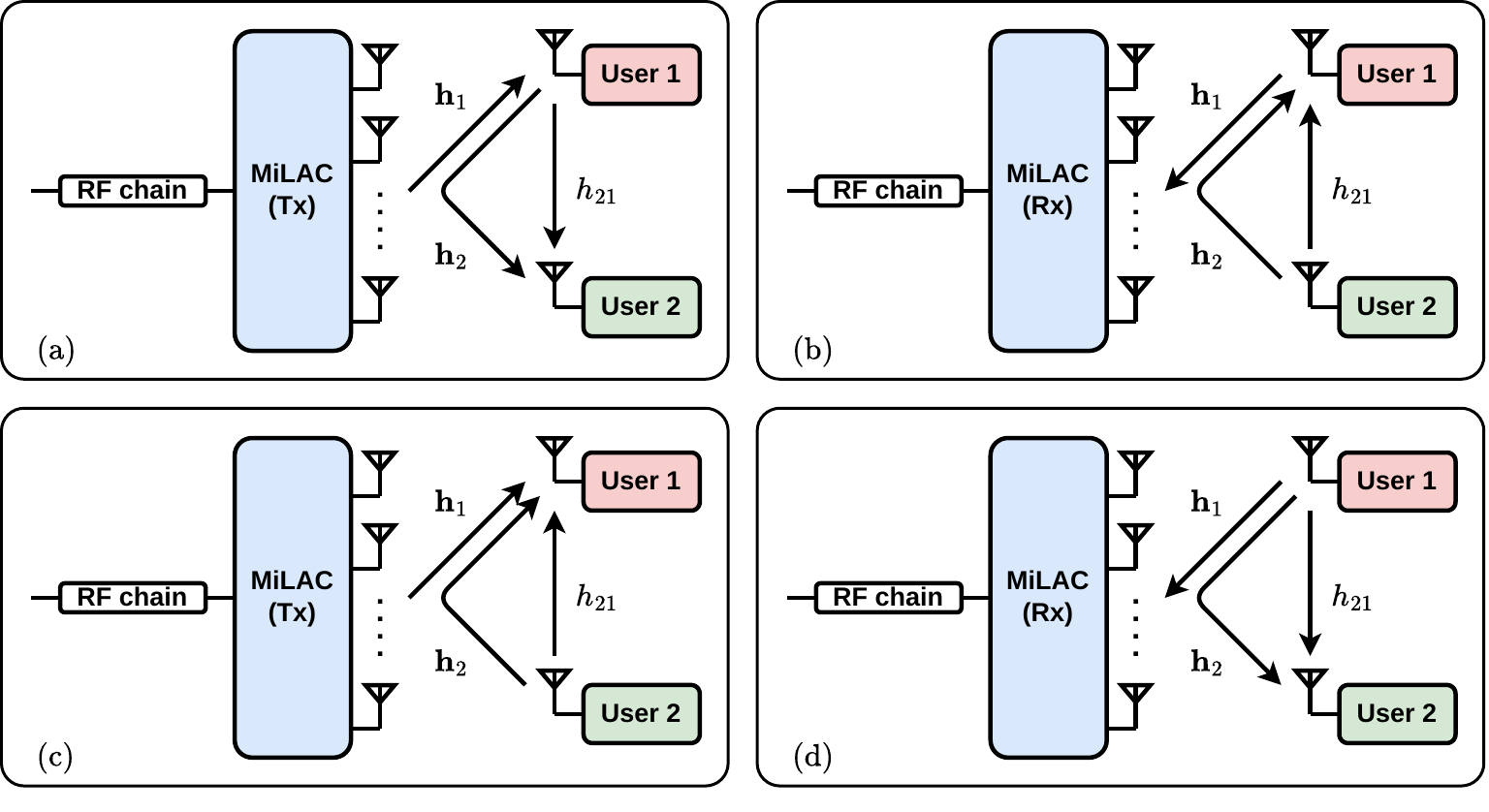}
\caption{Four wireless systems where MiLAC performs simultaneous active and passive beamforming.}
\label{fig:system}
\end{figure*}

To illustrate the utility of this dual functionality, four exemplary wireless scenarios are shown in Fig.~\ref{fig:system}.
In Fig.~\ref{fig:system}(a), the \gls{milac} transmits a signal to User~1, which operates in full-duplex and in turn transmits to User~2.
Therefore, the \gls{milac} simultaneously performs active beamforming toward User~1 and passive beamforming to help the communication from User~1 to 2.
In Fig.~\ref{fig:system}(b), the \gls{milac} receives the signal from User~1, which operates in full-duplex and in turn receives from User~2.
Here, the \gls{milac} simultaneously performs active beamforming toward User~1 and passive beamforming to aid the communication from User~2 to 1.
In Fig.~\ref{fig:system}(c), the \gls{milac} transmits a signal to User~1, which also receives the signal from User~2.
Hence, the \gls{milac} simultaneously performs active beamforming toward User~1 and passive beamforming to support the communication from User~2 to 1.
In Fig.~\ref{fig:system}(d), the \gls{milac} receives the signal from User~1, which also transmits to User~2.
Here, the \gls{milac} simultaneously performs active beamforming toward User~1 and passive beamforming to aid the communication from User~1 to 2.\footnote{Although the full-duplex scenarios in Fig.~\ref{fig:system} resemble the scenario studied in \cite{li25} to highlight the benefits of non-reciprocal \gls{bd-ris}, there is a fundamental difference.
Here, one stream is transmitted through the \gls{milac} and the other is reflected, while in \cite{li25} both streams are reflected by the \gls{bd-ris}.
As a result, a reciprocal \gls{milac} can achieve maximum performance here: none of the performance bounds derived in this paper can be improved by relaxing the symmetry constraints on the \gls{milac} scattering matrix.}

When a \gls{milac} is implemented with lossless and reciprocal hardware, its scattering matrix is unitary and symmetric.
This couples the transmissive and reflective capabilities of the \gls{milac}, creating a trade-off between the rates achievable by active and passive beamforming.
In this work, we explore the fundamental limits of this trade-off and make the following contributions.
\textit{First}, we establish a system model for dual-functional \gls{milac}-aided transceivers and provide the expressions of the active and passive achievable rates.
\textit{Second}, we derive tight upper bounds on both rates, and show that they cannot be achieved simultaneously.
Therefore, we characterize in closed form their capacity region, i.e., the set of all rate pairs that are simultaneously achievable, and the sum-rate capacity.
\textit{Third}, we provide a global-optimal solution to optimize the \gls{milac} to operate on the frontier of the capacity region and achieve the sum-rate capacity, and report simulation results that confirm the theoretical insights.

\section{System Model}

To showcase the capability of \gls{milac} to support active and passive beamforming simultaneously, we consider the communication system represented in Fig.~\ref{fig:system}(a) throughout the work.
In Fig.~\ref{fig:system}(a), an $N$-antenna transmitter equipped with a \gls{milac} communicates with a single-antenna user, denoted as User~1.
At the same time, User~1 operates in full-duplex mode and transmits to a second single-antenna user, denoted as User~2.
The objective is to reconfigure the \gls{milac} such that it simultaneously performs beamforming toward User~1 (active beamforming), and reflects the signal transmitted by User~1 toward User~2 (passive beamforming).

Since the \gls{milac}-aided transmitter sends a single stream of information, it only has a single \gls{rf} chain \cite{ner25-3}.
Therefore, the \gls{milac} microwave network has one input port and $N$ output ports, for a total of $N+1$ ports, and can be characterized by its scattering matrix $\bar{\boldsymbol{\Theta}}\in\mathbb{C}^{(N+1)\times(N+1)}$ \cite{poz11}.
Under the typical assumption that the \gls{milac} microwave network is implemented with lossless and reciprocal components, its scattering matrix is constrained to be unitary and symmetric, i.e., $\bar{\boldsymbol{\Theta}}^H\bar{\boldsymbol{\Theta}}=\mathbf{I}$ and $\bar{\boldsymbol{\Theta}}=\bar{\boldsymbol{\Theta}}^T$ \cite{poz11}.
Furthermore, it can be conveniently partitioned as
\begin{equation}
\bar{\boldsymbol{\Theta}}=
\begin{bmatrix}
\theta & \boldsymbol{\theta}^T\\
\boldsymbol{\theta} & \boldsymbol{\Theta}
\end{bmatrix},\label{eq:T}
\end{equation}
where $\theta\in\mathbb{C}$ is the reflection coefficient at port~1, $\boldsymbol{\theta}\in\mathbb{C}^{N\times1}$ is the transmission scattering matrix from port~1 to the $N$ ports connected to the $N$ antennas, and $\boldsymbol{\Theta}\in\mathbb{C}^{N\times N}$.
Physically, \eqref{eq:T} means that the signal reflected at the $N$ output ports of the \gls{milac} is the superposition of the signal sent by the \gls{rf} chain (processed through $\boldsymbol{\theta}$) and the signal impinging at the antennas (processed through $\boldsymbol{\Theta}$).
This superposition principle is what enables the \gls{milac} to simultaneously transmit the signal from the \gls{rf} chain and reflect the incident signal.
We now illustrate how $\bar{\boldsymbol{\Theta}}$ affects the communication between \gls{milac} and User~1, and between User~1 and 2.
In this work, all signals are \gls{rf} narrowband signals and are represented through their complex baseband equivalent, and frequency-selective effects are neglected to obtain fundamental insights.

Denote as $s_M\in\mathbb{C}$ the symbol sent by the \gls{milac}-aided transmitter, i.e., at the output of its \gls{rf} chain, and as $y_1\in\mathbb{C}$ the corresponding received signal at User~1.
Following \cite{ner25-3}, the signal $y_1$ is given by
\begin{equation}
y_1=\mathbf{h}_1^T\boldsymbol{\theta}x_M+n_1,\label{eq:y1}
\end{equation}
where $\mathbf{h}_1\in\mathbb{C}^{N\times 1}$ is the channel from User~1 and the \gls{milac}, $x_M=s_M/2$, and $n_1\sim\mathcal{CN}(0,\sigma_1^2)$ is the \gls{awgn} at User~1 with power $\sigma_1^2$.
In \eqref{eq:y1}, we have assumed perfect self-interference cancellation at User~1.
To understand how the \gls{milac} scattering matrix $\bar{\boldsymbol{\Theta}}$ impacts $y_1$, we equivalently rewrite \eqref{eq:y1} as
\begin{equation}
y_1=\bar{\mathbf{h}}_1^T\bar{\boldsymbol{\Theta}}\mathbf{b}x_M+n_1,\label{eq:y1-bar}
\end{equation}
where $\bar{\mathbf{h}}_1=[0,\mathbf{h}_1^T]^T\in\mathbb{C}^{(N+1)\times1}$ and $\mathbf{b}=[1,0,\dots,0]^T\in\mathbb{C}^{(N+1)\times1}$, where we explicitly highlight $\bar{\boldsymbol{\Theta}}$.
Following \eqref{eq:y1} and \eqref{eq:y1-bar}, the rate achieved between the \gls{milac} and User~1 is
\begin{align}
R_{1}
&=\log_2\left(1+P_M\frac{\left\vert\mathbf{h}_1^T\boldsymbol{\theta}\right\vert^2}{\sigma_1^2}\right)\label{eq:R1}\\
&=\log_2\left(1+P_M\frac{\left\vert\bar{\mathbf{h}}_1^T\bar{\boldsymbol{\Theta}}\mathbf{b}\right\vert^2}{\sigma_1^2}\right),\label{eq:R1-bar}
\end{align}
where $P_M=\mathbb{E}[\vert x_M\vert^2]=\mathbb{E}[\vert s_M\vert^2]/4$ is the transmit power.

Similarly, denote as $x_1\in\mathbb{C}$ the signal transmitted by User~1 and as $y_2\in\mathbb{C}$ the corresponding received signal at User~2.
The signal $y_2$ writes as
\begin{equation}
y_2=\left(h_{21}+\mathbf{h}_2^T\boldsymbol{\Theta}\mathbf{h}_1\right)x_1+n_2,\label{eq:y2}
\end{equation}
where $h_{21}\in\mathbb{C}$ is the direct channel from User~1 to 2, $\mathbf{h}_2\in\mathbb{C}^{N\times 1}$ is the channel from User~2 to the \gls{milac}, and $n_2\sim\mathcal{CN}(0,\sigma_2^2)$ is the \gls{awgn} at User~2 with power $\sigma_2^2$.
To highlight the role of the \gls{milac} scattering matrix $\bar{\boldsymbol{\Theta}}$, \eqref{eq:y2} can also be written as
\begin{equation}
y_2=\left(h_{21}+\bar{\mathbf{h}}_2^T\bar{\boldsymbol{\Theta}}\bar{\mathbf{h}}_1\right)x_1+n_2,\label{eq:y2-bar}
\end{equation}
where $\bar{\mathbf{h}}_2=[0,\mathbf{h}_2^T]^T\in\mathbb{C}^{(N+1)\times1}$.
From \eqref{eq:y2} and \eqref{eq:y2-bar}, the rate achieved between User~1 and 2 is
\begin{align}
R_{2}
&=\log_2\left(1+P_1\frac{\left\vert h_{21}+\mathbf{h}_2^T\boldsymbol{\Theta}\mathbf{h}_1\right\vert^2}{\sigma_2^2}\right)\label{eq:R2}\\
&=\log_2\left(1+P_1\frac{\left\vert h_{21}+\bar{\mathbf{h}}_2^T\bar{\boldsymbol{\Theta}}\bar{\mathbf{h}}_1\right\vert^2}{\sigma_2^2}\right),\label{eq:R2-bar}
\end{align}
where $P_1=\mathbb{E}[\vert x_1\vert^2]$ is the transmit power.

\section{MiLAC Optimization}
\label{sec:opt}

When reconfiguring the \gls{milac}, our objective is to optimize $\bar{\boldsymbol{\Theta}}$ subject to the unitary and symmetric constraints to simultaneously maximize the two rates $R_{1}$ and $R_{2}$.
In this section, we first optimize $\bar{\boldsymbol{\Theta}}$ to individually maximize $R_{1}$ and $R_{2}$.
Second, we characterize the capacity region, i.e., the set of all rate pairs $(R_{1},R_{2})$ that are simultaneously achievable.
Third, we maximize the sum-rate $R=R_{1}+R_{2}$.
To optimally design the \gls{milac} we assume perfect knowledge of the channels $\mathbf{h}_1$, $\mathbf{h}_2$, and $h_{21}$.\footnote{Differently from a \gls{ris}-aided system, the considered \gls{milac}-aided system has \gls{rf} chains at the \gls{milac}, which can receive and transmit signals, including pilot signals.
Therefore, the channels $\mathbf{h}_1$ and $\mathbf{h}_2$ can be estimated individually as in a multi-user \gls{milac}-aided system.}

\subsection{Individual Rate Maximization}
\label{sec:opt-ind}

Two opposite working modes can be identified for the \gls{milac}.
In the first mode, it is used purely for active beamforming, i.e., to precode the symbol $x_M$ for User~1.
In this case, we want to maximize $R_{1}$, which is upper bounded by
\begin{equation}
\bar{R}_{1}=\log_2\left(1+P_M\frac{\left\Vert\mathbf{h}_1\right\Vert^2}{\sigma_1^2}\right),\label{eq:R1-UB}
\end{equation}
according to the Cauchy-Schwarz inequality applied to \eqref{eq:R1} and because $\Vert\boldsymbol{\theta}\Vert\leq1$.
This upper bound can be exactly achieved if and only if $\boldsymbol{\theta}$ is designed such that
\begin{equation}
\frac{\mathbf{h}_1^T}{\left\Vert\mathbf{h}_1\right\Vert}\boldsymbol{\theta}=1,\label{eq:cond-1}
\end{equation}
as it can be noticed from \eqref{eq:R1}, or, equivalently, if and only if
\begin{equation}
\bar{\boldsymbol{\Theta}}\frac{\bar{\mathbf{h}}_1}{\left\Vert\mathbf{h}_1\right\Vert}=\mathbf{b},\label{eq:cond-1-bar}
\end{equation}
as it can be seen from \eqref{eq:R1-bar}.
Condition \eqref{eq:cond-1-bar} can be satisfied by optimizing the unitary and symmetric matrix $\bar{\boldsymbol{\Theta}}$ with \cite[Alg.~1]{ner24}, which was developed for \gls{bd-ris}, as a function of $\bar{\mathbf{h}}_1/\Vert\mathbf{h}_1\Vert$ and $\mathbf{b}$.

In the second mode, the \gls{milac} is used purely for passive beamforming, i.e., to enhance the link between User~1 and 2.
Here, we want to maximize $R_{2}$, which is upper bounded by
\begin{equation}
\bar{R}_{2}=\log_2\left(1+P_1\frac{\left(\left\vert h_{21}\right\vert+\left\Vert\mathbf{h}_2\right\Vert\left\Vert\mathbf{h}_1\right\Vert\right)^2}{\sigma_2^2}\right),\label{eq:R2-UB}
\end{equation}
following the Cauchy-Schwarz and triangle inequalities applied to \eqref{eq:R2} and because $\Vert\bar{\boldsymbol{\Theta}}\Vert\leq1$.
This upper bound can be exactly achieved if and only if $\boldsymbol{\Theta}$ is designed such that
\begin{equation}
\boldsymbol{\Theta}\frac{\mathbf{h}_1}{\left\Vert\mathbf{h}_1\right\Vert}=\frac{\mathbf{h}_2^*}{\left\Vert\mathbf{h}_2\right\Vert}e^{j\arg\left(h_{21}\right)},\label{eq:cond-2}
\end{equation}
as it can be noticed from \eqref{eq:R2}, or, equivalently, if and only if
\begin{equation}
\bar{\boldsymbol{\Theta}}\frac{\bar{\mathbf{h}}_1}{\left\Vert\mathbf{h}_1\right\Vert}=\frac{\bar{\mathbf{h}}_2^*}{\left\Vert\mathbf{h}_2\right\Vert}e^{j\arg\left(h_{21}\right)},\label{eq:cond-2-bar}
\end{equation}
as it can be seen from \eqref{eq:R2-bar}.
As for \eqref{eq:cond-1-bar}, condition \eqref{eq:cond-2-bar} can be achieved by optimizing $\bar{\boldsymbol{\Theta}}$ with \cite[Alg.~1]{ner24}.

We have derived sufficient and necessary conditions on $\bar{\boldsymbol{\Theta}}$ to maximize the two rates $R_{1}$ and $R_{2}$ in \eqref{eq:cond-1-bar} and \eqref{eq:cond-2-bar}, respectively.
These two conditions have the same left-hand side but orthogonal right-hand side (recall that $\mathbf{b}$ is zero everywhere except in its first entry, while $\bar{\mathbf{h}}_2$ has a zero first entry).
Thus, the two rates $R_{1}$ and $R_{2}$ cannot be simultaneously maximized.
If $R_{1}$ is maximized, i.e., $R_{1}=\bar{R}_{1}$, then $R_{2}=\underline{R}_{2}$, with 
\begin{equation}
\underline{R}_{2}=\log_2\left(1+P_1\frac{\left\vert h_{21}\right\vert^2}{\sigma_2^2}\right),
\end{equation}
which is obtained by substituting \eqref{eq:cond-1-bar} into \eqref{eq:R2-bar} and is the rate $R_{2}$ achievable by only using the direct link $h_{21}$ between User~1 and 2.
If $R_{2}$ is maximized, i.e., $R_{2}=\bar{R}_{2}$, then $R_{1}=0$, as obtained by plugging \eqref{eq:cond-2-bar} into \eqref{eq:R1-bar}.

\subsection{Capacity Region Characterization}
\label{sec:opt-cap}

Since $R_{1}$ and $R_{2}$ cannot be simultaneously maximized, the sum-rate upper bound $\bar{R}=\bar{R}_{1}+\bar{R}_{2}$ cannot be reached.
Thus, there exists a trade-off between the rates $R_{1}$ and $R_{2}$ that can be achieved.
Formally, this trade-off is given by the capacity region $\mathcal{C}$, i.e., the set of all rate pairs $(R_{1},R_{2})$ that are simultaneously achievable.
In the following, we characterize the frontier of such a region.

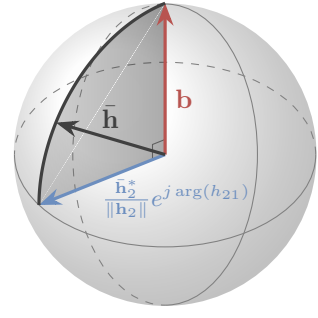
\begin{figure}
\centering
\begin{tikzpicture}[scale=1, transform shape]
\shade[ball color=black!2, opacity=0.4] (0,0) circle (2);

\draw[gray, thin, dashed]
(1.99,0) arc[start angle=0, end angle=180, x radius=1.99, y radius=1.22];

\draw[gray, thin]
(-1.99,0) arc[start angle=180, end angle=360, x radius=1.99, y radius=1.22];

\draw[gray, thin, dashed]
(0,1.99) arc[start angle=90, end angle=270, x radius=1.22, y radius=1.99];

\draw[gray, thin]
(0,-1.99) arc[start angle=-90, end angle=90, x radius=1.22, y radius=1.99];

\draw[gray, thin]
(0,0.2,0) -- (-0.1,0.2,0.1732) -- (-0.1,0,0.1732);

\fill[shade, left color=black!2, right color=black!80, opacity=0.2]
(0,0,0) -- (0,2,0) -- (-1,0,1.732)
plot[domain=90:0, samples=60] ({-sin(\x)},{2*cos(\x)},{1.732*sin(\x)})
-- cycle;

\draw[very thick,-{Stealth}, myred]
(0,0,0) -- (0,2,0)
node[midway, below right] {$\mathbf{b}$};
    
\draw[very thick,-{Stealth}, myblue]
(0,0,0) -- (-1,0,1.732)
node[font=\small, midway, below right=-4pt] {$\frac{\bar{\mathbf{h}}_2^*}{\Vert\mathbf{h}_2\Vert}e^{j\arg(h_{21})}$};
    
\draw[very thick,-{Stealth}, darkgray]
(0,0,0) -- (-0.866,1,1.5)
node[midway, above] {$\bar{\mathbf{h}}$};

\draw[very thick, darkgray]
plot[domain=0:90, samples=100] ({-sin(\x)},{2*cos(\x)},{1.732*sin(\x)});
\end{tikzpicture}
\caption{The unit $N$-sphere, the orthonormal vectors $\bar{\mathbf{h}}_2^*/\left\Vert\mathbf{h}_2\right\Vert e^{j\arg(h_{21})}$ and $\mathbf{b}=[1,0,\dots,0]^T$, and the shortest arc connecting them.}
\label{fig:representation}
\end{figure}

We begin by giving an intuition on how the \gls{milac} should be reconfigured to operate on the frontier of the capacity region, and then provide a formal proof to show that this intuitive design is optimal.
The rates $R_{1}$ and $R_{2}$ are individually maximized if conditions \eqref{eq:cond-1-bar} and \eqref{eq:cond-2-bar} are fulfilled, respectively.
Since it is impossible to fulfill both of them simultaneously, we can impose a compromise
\begin{equation}
\bar{\boldsymbol{\Theta}}\frac{\bar{\mathbf{h}}_1}{\left\Vert\mathbf{h}_1\right\Vert}=\bar{\mathbf{h}},\label{eq:cond-opt}
\end{equation}
where $\bar{\mathbf{h}}\in\mathbb{C}^{(N+1)\times1}$ is introduced as
\begin{equation}
\bar{\mathbf{h}}=t\frac{\bar{\mathbf{h}}_2^*}{\left\Vert\mathbf{h}_2\right\Vert}e^{j\arg\left(h_{21}\right)}+\sqrt{1-t^2}\mathbf{b},
\end{equation}
depending on a parameter $t\in[0,1]$.
Note that $\bar{\mathbf{h}}$ is a unit-norm vector that always lies on the shortest arc of the unit $N$-sphere connecting $\bar{\mathbf{h}}_2^*/\Vert\mathbf{h}_2\Vert e^{j\arg(h_{21})}$ and $\mathbf{b}$, as illustrated in Fig.~\ref{fig:representation}.
Since $\bar{\mathbf{h}}$ has unit norm, condition \eqref{eq:cond-opt} can always be fulfilled by optimizing $\bar{\boldsymbol{\Theta}}$ with \cite[Alg.~1]{ner24}.
For $t=0$, condition \eqref{eq:cond-opt} boils down to \eqref{eq:cond-1-bar} and $R_{1}$ is maximized; for $t=1$, condition \eqref{eq:cond-opt} boils down to \eqref{eq:cond-2-bar} and $R_{2}$ is maximized; and for other values of $t$, we achieve
\begin{align}
R_{1}(t)&=\log_2\left(1+P_M\frac{\left\Vert\mathbf{h}_1\right\Vert^2\left(1-t^2\right)}{\sigma_1^2}\right),\label{eq:R1-curve}\\
R_{2}(t)&=\log_2\left(1+P_1\frac{\left(\left\vert h_{21}\right\vert+\left\Vert\mathbf{h}_2\right\Vert\left\Vert\mathbf{h}_1\right\Vert t\right)^2}{\sigma_2^2}\right),\label{eq:R2-curve}
\end{align}
which can be shown by substituting \eqref{eq:cond-opt} into \eqref{eq:R1-bar} and \eqref{eq:R2-bar}.
Interestingly, these two parametric expressions of $R_{1}(t)$ and $R_{2}(t)$ identify the frontier of the capacity region, as formalized in the following proposition.
\begin{proposition}
The frontier of the capacity region containing all possible pairs $(R_{1},R_{2})$ is given by the curve identified by the parametric expressions \eqref{eq:R1-curve} and \eqref{eq:R2-curve}, for $t\in[0,1]$.
\label{pro:curve}
\end{proposition}
\begin{proof}
We prove the proposition by showing that: \textit{i)} if $R_{2}(t)$ is achieved for some $t\in[0,1]$, then $R_{1}\leq R_{1}(t)$, and \textit{ii)} if $R_{1}(t)$ is achieved for some $t\in[0,1]$, then $R_{2}\leq R_{2}(t)$.

First, achieving $R_{2}(t)$ implies that
\begin{equation}
\left\vert\hat{\mathbf{h}}_1^T\hat{\mathbf{h}}_2\right\vert=t,\label{eq:t}
\end{equation}
where we introduced $\hat{\mathbf{h}}_1=\bar{\boldsymbol{\Theta}}\bar{\mathbf{h}}_1/\Vert\mathbf{h}_1\Vert$ and $\hat{\mathbf{h}}_2=\bar{\mathbf{h}}_2/\Vert\mathbf{h}_2\Vert$.
We now use the Parseval's identity to express $\Vert\hat{\mathbf{h}}_1\Vert^2$ as
\begin{equation}
\left\Vert\hat{\mathbf{h}}_1\right\Vert^2=\sum_{\mathbf{v}\in\mathcal{B}}\left\vert\hat{\mathbf{h}}_1^T\mathbf{v}\right\vert^2,\label{eq:parseval}
\end{equation}
where $\mathcal{B}=\{\mathbf{v}_{1},\ldots,\mathbf{v}_{N+1}\}$ is any orthonormal basis of $\mathbb{C}^{N+1}$, with $\mathbf{v}_{n}\in\mathbb{C}^{N+1\times 1}$ for $n=1,\ldots,N+1$.
We construct such a basis $\mathcal{B}$ as
\begin{equation}
\mathcal{B}=\left\{\hat{\mathbf{h}}_2,\mathbf{b},\mathbf{v}_{3},\ldots,\mathbf{v}_{N+1}\right\},\label{eq:B}
\end{equation}
where $\mathbf{v}_{3},\ldots,\mathbf{v}_{N+1}$ are orthonormal vectors orthogonal to $\hat{\mathbf{h}}_2$ and $\mathbf{b}$.
Note that \eqref{eq:B} is a valid basis since $\hat{\mathbf{h}}_2$ and $\mathbf{b}$ are orthonormal.
Recalling that $\Vert\hat{\mathbf{h}}_1\Vert^2=1$ and applying the Parseval's identity to rewrite $\Vert\hat{\mathbf{h}}_1\Vert^2$ as in \eqref{eq:parseval}, we obtain
\begin{align}
1
&=\left\vert\hat{\mathbf{h}}_1^T\hat{\mathbf{h}}_2\right\vert^2+\left\vert\hat{\mathbf{h}}_1^T\mathbf{b}\right\vert^2+\sum_{n=3}^{N+1}\left\vert\hat{\mathbf{h}}_1^T\mathbf{v}_n\right\vert^2\\
&\geq\left\vert\hat{\mathbf{h}}_1^T\hat{\mathbf{h}}_2\right\vert^2+\left\vert\hat{\mathbf{h}}_1^T\mathbf{b}\right\vert^2.\label{eq:norm-ineq}
\end{align}
Substituting \eqref{eq:t} into \eqref{eq:norm-ineq}, we obtain $\vert\hat{\mathbf{h}}_1^T\mathbf{b}\vert\leq\sqrt{1-t^2}$, giving $R_{1}\leq R_{1}(t)$.
Second, achieving $R_{1}(t)$ implies that
\begin{equation}
\left\vert\hat{\mathbf{h}}_1^T\mathbf{b}\right\vert=\sqrt{1-t^2}.\label{eq:1-t}
\end{equation}
Similarly to the previous discussion, substituting \eqref{eq:1-t} into \eqref{eq:norm-ineq}, we obtain $\vert\hat{\mathbf{h}}_1^T\hat{\mathbf{h}}_2\vert\leq t$, giving $R_{2}\leq R_{2}(t)$.
\end{proof}

We can reach all points on the capacity region frontier depending on the parameter $t\in[0,1]$.
If $t=0$, the \gls{milac} is used purely for active beamforming (as a transmitter); if $t=1$, it is used purely for passive beamforming (as a \gls{bd-ris}); and if $t\in(0,1)$, it is used for active and passive beamforming simultaneously (which is reminiscent of a \gls{star-ris}).

\subsection{Sum-Rate Maximization}
\label{sec:opt-sum}

Among all the points on the capacity region, we are interested in the sum-rate capacity $C$ defined as the maximum achievable sum of rates $C=\max_{(R_{1},R_{2})\in\mathcal{C}}R_{1}+R_{2}$.
Since the sum-rate capacity lies on the capacity region frontier given by \eqref{eq:R1-curve} and \eqref{eq:R2-curve}, we want to find the optimal value of $t\in[0,1]$, called $t^\star$, that maximizes the sum-rate $R(t)=R_{1}(t)+R_{2}(t)$.
To this end, we find all stationary points of $R(t)$ for $t\in[0,1]$.
It can be shown that solving $dR(t)/dt=0$ is equivalent to finding the roots of the cubic equation
\begin{equation}
at^3+bt^2+ct+d=0,\label{eq:cubic}
\end{equation}
where
\begin{align}
a=&2P_1P_M\left\Vert\mathbf{h}_1\right\Vert^3\left\Vert\mathbf{h}_2\right\Vert^2\\
b=&3P_1P_M\left\Vert\mathbf{h}_1\right\Vert^2\left\Vert\mathbf{h}_2\right\Vert\left\vert h_{21}\right\vert\\
c=&P_M\left\Vert\mathbf{h}_1\right\Vert\left(P_1\left\vert h_{21}\right\vert^2+\sigma_2^2\right)\\
-&P_1\left\Vert\mathbf{h}_1\right\Vert\left\Vert\mathbf{h}_2\right\Vert^2\left(P_M\left\Vert\mathbf{h}_1\right\Vert^2+\sigma_1^2\right)\\
d=&-P_1\left\Vert\mathbf{h}_2\right\Vert\left\vert h_{21}\right\vert\left(P_M\left\Vert\mathbf{h}_1\right\Vert^2+\sigma_1^2\right).
\end{align}
Thus, there are always up to five candidate values for $t^\star$, given by the real roots of \eqref{eq:cubic} lying in the interval $[0,1]$, which are up to three, together with the extreme values $0$ and $1$.
After solving \eqref{eq:cubic}, the optimal $t^\star$ can be readily found by evaluating $R(t)$ in these candidate points.

In the case the direct link between User~1 and 2 is highly obstructed, i.e., $h_{21}=0$, \eqref{eq:cubic} boils down to $t(at^2+c)=0$ since $b=0$ and $d=0$.
This equation has at most two real solutions in the interval $[0,1]$, namely $t=0$ and
\begin{equation}
t=\sqrt{\frac{-c}{a}}=\sqrt{\frac{P_1\left\Vert\mathbf{h}_2\right\Vert^2\left(P_M\left\Vert\mathbf{h}_1\right\Vert^2+\sigma_1^2\right)-P_M\sigma_2^2}
{2P_1P_M\left\Vert\mathbf{h}_1\right\Vert^2\left\Vert\mathbf{h}_2\right\Vert^2}},\label{eq:2}
\end{equation}
the second of which could not be real.
If $h_{21}=0$ and we are in the high-\gls{snr} regime, i.e., $P_M\left\Vert\mathbf{h}_1\right\Vert^2\gg\sigma_1^2$ and $P_1\left\Vert\mathbf{h}_1\right\Vert^2\left\Vert\mathbf{h}_2\right\Vert^2\gg\sigma_2^2$, then \eqref{eq:2} simplifies to $t\approx\sqrt{1/2}$.
Therefore, in this case, there are three candidate values for $t^\star$, namely $t\in\{0,\sqrt{1/2},1\}$.
Comparing $R(0)$, $R(\sqrt{1/2})$, and $R(1)$ when $h_{21}=0$ and we are at high-\gls{snr} regime, it is simple to show that the optimal $t$ is $t^\star=\sqrt{1/2}$.
This is aligned with the fact that at high-\gls{snr} uniform power allocation is optimal and offers a multiplexing gain of 2, rather than 1 (which is offered by the extreme $t=0$ and $t=1$).
\section{Numerical Results}
\label{sec:results}

\begin{figure}[t]
\centering
\subfigure[]{
\includegraphics[width=0.234\textwidth]{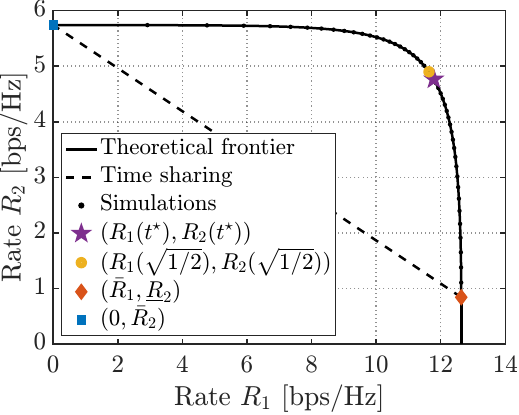}}
\subfigure[]{
\includegraphics[width=0.234\textwidth]{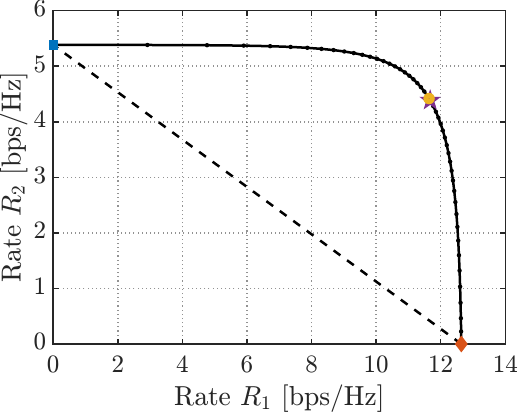}}
\caption{Capacity region (a) with and (b) without direct link $h_{21}$.}
\label{fig:region}
\end{figure}

\begin{figure}[t]
\centering
\subfigure[]{
\includegraphics[width=0.234\textwidth]{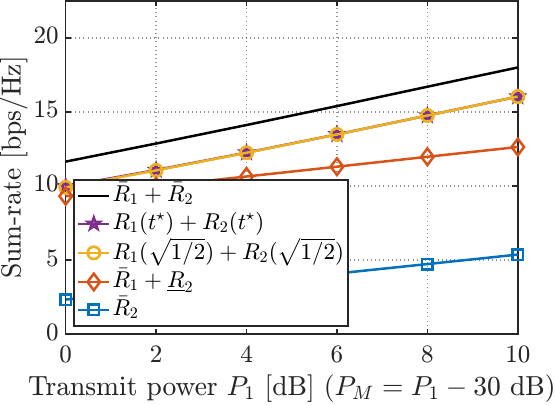}}
\subfigure[]{
\includegraphics[width=0.234\textwidth]{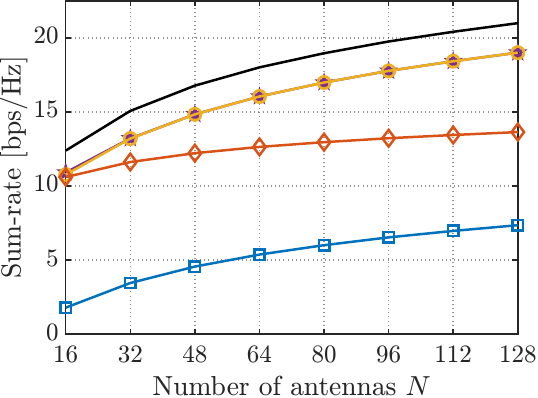}}
\caption{Sum-rate versus (a) the transmit power $P_1$ and (b) the number of antennas $N$ at the MiLAC, without direct link $h_{21}$.}
\label{fig:sum-rate}
\end{figure}

In this section, we provide numerical results to validate our theoretical insights.
We consider \gls{iid} Rayleigh distributed channels, i.e., $h_{21}\sim\mathcal{CN}(0,g_{21})$, $\mathbf{h}_1\sim\mathcal{CN}(\mathbf{0},g_1\mathbf{I})$, and $\mathbf{h}_2\sim\mathcal{CN}(\mathbf{0},g_2\mathbf{I})$, with path gains $g_{21}=10^{-12}$, and $g_1=10^{-7}$, and $g_2=10^{-7}$.
We set the number of antennas at the \gls{milac} to $N=64$, the transmit powers to $P_1=10$~dB and $P_M=10$~dBm, and the noise powers to $\sigma_1^2=\sigma_2^2=-80$~dBm.

Fig.~\ref{fig:region} shows the capacity region of the achievable pairs $(R_1,R_2)$.
The corner points $(\bar{R}_{1},\underline{R}_{2})$ and $(0,\bar{R}_{2})$ are achieved by using the \gls{milac} purely for active or passive beamforming, respectively, as described in Section~\ref{sec:opt-ind}.
By time-sharing between these two operating modes, all points on the dashed line can be achieved, obtaining a multiplexing gain of 1 only.
Significantly better performance can be achieved by reconfiguring the \gls{milac} according to the approach in Section~\ref{sec:opt-cap}.
By reconfiguring the \gls{milac} to satisfy condition \eqref{eq:cond-opt} for $t\in[0,1]$, all points on the capacity region frontier can be achieved, as described by the parametric expressions \eqref{eq:R1-curve} and \eqref{eq:R2-curve}.
The sum-rate capacity is achieved with $t=t^\star$, on the point $(R_{1}(t^\star),R_{2}(t^\star))$, as described in Section~\ref{sec:opt-sum}.
Since, $t^\star\approx\sqrt{1/2}$ (as theoretically proved for when $h_{21}=0$ and we are in the high-\gls{snr} regime), the point $(R_{1}(t^\star),R_{2}(t^\star))$ approximately corresponds to $(R_{1}(\sqrt{1/2}),R_{2}(\sqrt{1/2}))$.

The capacity region frontier characterized in Fig.~\ref{fig:region} is valid for all four scenarios represented in Fig.~\ref{fig:system}.
For the scenarios in Fig.~\ref{fig:system}(a) and (b), the frontier can be tightly achieved if User~1 perfectly cancels self-interference, while for the scenarios in Fig.~\ref{fig:system}(c) and (d), it can be tightly achieved if the two transmitted signals do not interfere (e.g., they are orthogonal in frequency).
If these assumptions are not satisfied, the derived capacity region frontier provides an upper bound on the actual capacity region.

Fig.~\ref{fig:sum-rate} shows the sum-rate $R=R_1+R_2$ obtained by operating at specific values of the parameter $t$, compared to the sum-rate upper bound $\bar{R}=\bar{R}_{1}+\bar{R}_{2}$.
The curves $\bar{R}_{1}+\underline{R}_{2}$ and $\bar{R}_{2}$ are the sum-rates achievable by using the \gls{milac} purely for active or passive beamforming, respectively, i.e., with $t=0$ or $t=1$.
By simultaneously using the \gls{milac} for active and passive beamforming, the sum-rate capacity can be achieved with $t=t^\star$, which approximately corresponds to the sum-rate with $t=\sqrt{1/2}$.
Although the upper bound $\bar{R}$ cannot be achieved, it is nearly reached by the sum-rate capacity, which has a multiplexing gain of 2 unlike $\bar{R}_{1}+\underline{R}_{2}$ and $\bar{R}_{2}$.

\section{Conclusion}

We propose a dual-functionality framework for transceivers equipped with a \gls{milac}.
In this framework, a \gls{milac} can simultaneously execute active beamforming for transmission/reception and passive beamforming for reflecting external signals, as a \gls{bd-ris}.
Due to the unitary and symmetric constraints of lossless and reciprocal \glspl{milac}, a fundamental trade-off exists between these two operating modes.
We characterize the limits of this trade-off by deriving the system's capacity region and providing an optimal reconfiguration strategy to achieve sum-rate capacity.
Future research could extend the proposed optimization framework to account for hardware imperfections (such as losses and wideband behavior of the \gls{milac}) and to more general communication systems (such as with multi-antenna users).
In addition, how to efficiently acquire channel knowledge in the proposed system could be investigated.

\bibliographystyle{IEEEtran}
\bibliography{IEEEabrv,main}

\begin{thebibliography}{10}
\providecommand{\url}[1]{#1}
\csname url@samestyle\endcsname
\providecommand{\newblock}{\relax}
\providecommand{\bibinfo}[2]{#2}
\providecommand{\BIBentrySTDinterwordspacing}{\spaceskip=0pt\relax}
\providecommand{\BIBentryALTinterwordstretchfactor}{4}
\providecommand{\BIBentryALTinterwordspacing}{\spaceskip=\fontdimen2\font plus
\BIBentryALTinterwordstretchfactor\fontdimen3\font minus \fontdimen4\font\relax}
\providecommand{\BIBforeignlanguage}[2]{{%
\expandafter\ifx\csname l@#1\endcsname\relax
\typeout{** WARNING: IEEEtran.bst: No hyphenation pattern has been}%
\typeout{** loaded for the language `#1'. Using the pattern for}%
\typeout{** the default language instead.}%
\else
\language=\csname l@#1\endcsname
\fi
#2}}
\providecommand{\BIBdecl}{\relax}
\BIBdecl

\bibitem{ner25-2}
M.~Nerini and B.~Clerckx, ``Analog computing for signal processing and communications – {Part II}: Toward gigantic {MIMO} beamforming,'' \emph{IEEE Trans. Signal Process.}, vol.~73, pp. 5198--5212, 2025.

\bibitem{ner25-3}
M.~Nerini and B.~Clerckx, ``Capacity of {MIMO} systems aided by microwave linear analog computers ({MiLACs}),'' \emph{arXiv preprint arXiv:2506.05983}, 2025.

\bibitem{fan26}
T.~Fang, X.~Zhou, and Y.~Mao, ``On the performance of lossless reciprocal {MiLAC} architectures in multi-user networks,'' \emph{IEEE Wireless Commun. Lett.}, vol.~15, pp. 2609--2613, 2026.

\bibitem{wu26}
Z.~Wu, M.~Nerini, and B.~Clerckx, ``Microwave linear analog computer ({MiLAC})-aided multiuser {MISO}: Fundamental limits and beamforming design,'' \emph{arXiv preprint arXiv:2601.10060}, 2026.

\bibitem{zho26}
X.~Zhou, T.~Fang, Y.~Mao, and B.~Clerckx, ``Two-layer microwave linear analog computer ({MiLAC})-aided multi-user {MISO} networks,'' \emph{arXiv preprint arXiv:2604.24303}, 2026.

\bibitem{zha26}
Y.~Zhang, P.~Zheng, and T.~Y. Al-Naffouri, ``Quantization-aware {EE} optimization and {SE-EE} tradeoff for milac-aided {MU-MISO} beamforming,'' \emph{arXiv preprint arXiv:2604.24538}, 2026.

\bibitem{liu26}
Z.~Liu, Z.~Wu, and B.~Clerckx, ``Microwave linear analog computer ({MiLAC})-aided {MIMO} radar sensing: Transmit beamforming design and {DoA} estimation,'' \emph{arXiv preprint arXiv:2605.21020}, 2026.

\bibitem{wu19}
Q.~Wu and R.~Zhang, ``Intelligent reflecting surface enhanced wireless network via joint active and passive beamforming,'' \emph{IEEE Trans. Wireless Commun.}, vol.~18, no.~11, pp. 5394--5409, 2019.

\bibitem{wu21}
Q.~Wu, S.~Zhang, B.~Zheng, C.~You, and R.~Zhang, ``Intelligent reflecting surface-aided wireless communications: A tutorial,'' \emph{IEEE Trans. Commun.}, vol.~69, no.~5, pp. 3313--3351, 2021.

\bibitem{mu22}
X.~Mu, Y.~Liu, L.~Guo, J.~Lin, and R.~Schober, ``Simultaneously transmitting and reflecting ({STAR}) {RIS} aided wireless communications,'' \emph{IEEE Trans. Wireless Commun.}, vol.~21, no.~5, pp. 3083--3098, 2022.

\bibitem{she22}
S.~Shen, B.~Clerckx, and R.~Murch, ``Modeling and architecture design of reconfigurable intelligent surfaces using scattering parameter network analysis,'' \emph{IEEE Trans. Wireless Commun.}, vol.~21, no.~2, pp. 1229--1243, 2022.

\bibitem{li25}
H.~Li and B.~Clerckx, ``Non-reciprocal beyond diagonal {RIS}: Multiport network models and performance benefits in full-duplex systems,'' \emph{IEEE Trans. Commun.}, vol.~73, no.~11, pp. 12\,221--12\,234, 2025.

\bibitem{poz11}
D.~M. Pozar, \emph{Microwave engineering}.\hskip 1em plus 0.5em minus 0.4em\relax John Wiley \& Sons, 2011.

\bibitem{ner24}
M.~Nerini, S.~Shen, and B.~Clerckx, ``Closed-form global optimization of beyond diagonal reconfigurable intelligent surfaces,'' \emph{IEEE Trans. Wireless Commun.}, vol.~23, no.~2, pp. 1037--1051, 2024.

\end{thebibliography}

\end{document}